\documentclass{article}
\usepackage[utf8]{inputenc}

\usepackage{fullpage}
\usepackage{natbib}
\usepackage{amsmath}
\usepackage{amsthm}

\usepackage{cleveref}

\newtheorem{theorem}{Theorem}

\usepackage{amsmath}
\usepackage{amssymb}
\usepackage{amsfonts}
\usepackage{ifthen}
\usepackage{xcolor}

\usepackage{bbm}
\usepackage{BOONDOX-cal}
\newcommand{\composed}[3]{#1{#2{#3}}}
\newcommand{\double}[2]{\kern.15ex#1{\kern-.15ex#1{\kern-.15ex#2}}}

\newcommand{\themodifier}{}  % '','forval','forquant','forbid','ironed','constrained'

\newcommand{\ironed}{\bar}
\newcommand{\constrained}{\hat}

\newcommand{\thestage}{an}           % 'an','exante', 'expost','interim'

\newcommand{\interimfont}[1]{#1}

\newcommand{\interim}[1]{{\setstagefont{\interimfont}#1}}

\newcommand{\stagefont}{\mathrm}
\newcommand{\setstagefont}[1]{\renewcommand{\stagefont}{#1}}

\newcommand{\decoration}{\noaccents}

\newcommand{\stagemodifier}[2]{\composed%
{\ifthenelse{\equal{#2}{}}{}{\csname #2\endcsname}}%
{\ifthenelse{\equal{#1}{}}{}{\csname #1\endcsname}}%
}

\newcommand{\thestagemodifier}{\stagemodifier{\thestage}{\themodifier}}

\newcommand{\ensuredecoration}{\renewcommand{\decoration}{\thestagemodifier}}

\ensuredecoration

\newcommand{\optimized}[1]{#1\opt}
\newcommand{\differentiated}[1]{#1'}
\newcommand{\tagged}[2]{{#2}^{#1}}

\newcommand{\primedarg}[1]{#1\primed}
\newcommand{\noaccents}[1]{#1}

\newcommand{\additemsubscript}{}
\newcommand{\itemsubscript}{}

\newcommand{\drop}[1]{}
\newcommand{\definevariant}[3]{\expandafter\newcommand\expandafter{\csname \expandafter\drop\string#1#2\endcsname}{}
\expandafter\DeclareRobustCommand\expandafter{\csname \expandafter\drop\string#1#2\endcsname}{#3}}
\newcommand{\usevariant}[2]{\csname \expandafter\drop\string#1#2\endcsname}

\newcommand{\agind}[1][\agent]{_{\smash{#1\itemsubscript}}}
\newcommand{\agith}[1][\agent]{_{(\smash{#1})\smash{\itemsubscript}}}
\newcommand{\minusagind}[1][\agent]{_{\smash{-#1}}}

\newcommand{\newagentvar}[3][\decoration]{%
\definevariant{#2}{}{#1{\stagefont{#3}}\additemsubscript}%
\definevariant{#2}{i}{#1{\stagefont{#3}}\agind}
\definevariant{#2}{ith}{#1{\stagefont{#3}}\agith}
\definevariant{#2}{s}{\kern.1ex#1{\kern-.1ex\boldsymbol{\stagefont{#3}}}}
\definevariant{#2}{smi}{\kern.1ex#1{\kern-.1ex\boldsymbol{\stagefont{#3}}}\minusagind} 
}%

\newcommand{\itind}[1][\itm]{_{#1}}
\newcommand{\minusitind}[1][\itm]{_{-#1}}

\newcommand{\newitemvar}[3][\decoration]{%
\definevariant{#2}{}{#1{\stagefont{#3}}}%
\definevariant{#2}{j}{#1{\stagefont{#3}}\itind}
\definevariant{#2}{s}{\kern.1ex#1{\kern-.1ex\boldsymbol{\stagefont{#3}}}}
\definevariant{#2}{smi}{\kern.1ex#1{\kern-.1ex\boldsymbol{\stagefont{#3}}}\minusitind} 
}%

\newcommand{\newagentvarsm}[4]{\newagentvar[\stagemodifier{#3}{#4}]{#1}{#2}}

\newcommand{\theval}{v}
\newagentvarsm{\val}{\theval}{an}{none}
\newagentvarsm{\Val}{\theval}{interim}{forquant}
\newagentvarsm{\type}{t}{an}{none}
\newagentvarsm{\othertype}{s}{an}{none}

\newcommand{\therev}{R}
\newcommand{\thepricerev}{P}
\newagentvarsm{\rev}{\therev}{interim}{nomodifer}
\newagentvarsm{\marg}{\therev}{interim}{differentiated}
\newagentvarsm{\rawrev}{\thepricerev}{interim}{}
\newagentvarsm{\pricerev}{\thepricerev}{interim}{}
\newagentvarsm{\rawmarg}{\thepricerev}{interim}{differentiated}
\newagentvarsm{\pricemarg}{\thepricerev}{interim}{differentiated}

\newcommand{\thevirt}{\varphi}
\newcommand{\thecumvirt}{\Phi}
\newagentvarsm{\virt}{\thevirt}{interim}{}
\newagentvarsm{\cumvirt}{\thecumvirt}{interim}{}
\newagentvarsm{\qvirt}{\thevirt}{interim}{}
\newagentvarsm{\ivirt}{\thevirt}{interim}{ironed}
\newagentvarsm{\icumvirt}{\thecumvirt}{interim}{ironed}

\newagentvarsm{\dist}{F}{interim}{}
\newagentvarsm{\dens}{f}{interim}{}
\newagentvarsm{\hazard}{h}{interim}{}
\newagentvarsm{\cumhazard}{H}{interim}{}

\newcommand{\thealloc}{x}
\newagentvar{\alloc}{\thealloc}
\newagentvar[\ironed]{\ialloc}{\thealloc}

\newagentvarsm{\zee}{z}{an}{}
\newagentvarsm{\wye}{y}{an}{}
\newagentvarsm{\cee}{c}{an}{}
\newagentvarsm{\delt}{\delta}{an}{}

\newcommand{\thefalloc}{z}
\newagentvar{\falloc}{\thefalloc}
\newagentvarsm{\epfalloc}{\thefalloc}{expost}{}

\newcommand{\thequant}{q}
\newagentvarsm{\quant}{\thequant}{an}{none}
\newagentvarsm{\exquant}{\thequant}{an}{constrained}
\newagentvarsm{\critquant}{\thequant}{an}{constrained}  %remove?
\newagentvarsm{\monoq}{\thequant}{an}{optconstrained}
\newagentvar[\interim]{\toquant}{\thequant}

\newcommand{\theqalloc}{y}
\newcommand{\thecumalloc}{Y}
\newagentvarsm{\qalloc}{\theqalloc}{interim}{none}
\newagentvarsm{\cumalloc}{\thecumalloc}{interim}{none}
\newagentvarsm{\calloc}{\theqalloc}{interim}{constrained}
\newagentvarsm{\cumcalloc}{\thecumalloc}{interim}{constrained}
\newagentvarsm{\iqalloc}{\theqalloc}{interim}{ironed}
\newagentvarsm{\icumalloc}{\thecumalloc}{interim}{ironed}

\newagentvar{\qprice}{\price}
\newagentvar{\qrev}{R}
\newagentvar[\ironed]{\iqrev}{\qrev}

\newagentvarsm{\excalloc}{\theqalloc}{expost}{constrained}
\newagentvarsm{\exalloc}{\theqalloc}{expost}{forquantile}
\newagentvarsm{\extalloc}{\thealloc}{expost}{forvalue}
\newagentvarsm{\exfalloc}{\thefalloc}{expost}{forvalue}

\newagentvarsm{\inducedrev}{\therev}{interim}{primedarg}

\newagentvar[\tilde]{\pseudorev}{\therev}
\newagentvar[\tilde]{\pseudorawrev}{\thepricerev}

\newagentvar{\typespace}{{\cal T}}
\newagentvar{\typesubspace}{S}

\newcommand{\theoutcome}{w}
\newagentvarsm{\outcome}{\theoutcome}{an}{none}
\newagentvarsm{\toutcome}{\theoutcome}{interim}{forval}
\newagentvar{\outcomespace}{{\cal W}}

\newcommand{\served}[1]{#1^1}
\newcommand{\nonserved}[1]{#1^0}
\newcommand{\alloced}[1]{#1^{\alloc}}
\newcommand{\allocedi}[1]{#1^{\alloci}}
\newagentvar[\alloced]{\xoutcome}{\outcome}
\newagentvar[\allocedi]{\xioutcome}{\outcome}
\newagentvar[\served]{\soutcome}{\outcome}

\newagentvar[\nonserved]{\nsoutcome}{\outcome}

\newcommand{\theprice}{p}
\newagentvar{\price}{\theprice}

\newcommand{\thepay}{p}
\newagentvarsm{\pay}{\thepay}{an}{}
\newagentvarsm{\vpay}{\thepay}{interim}{forval}
\newagentvarsm{\tpay}{\thepay}{interim}{forval}
\newagentvarsm{\epvpay}{\thepay}{expost}{forval}
\newagentvarsm{\talloc}{\thealloc}{interim}{forval}
\newagentvarsm{\valloc}{\thealloc}{interim}{forval}
\newagentvarsm{\eptalloc}{\thealloc}{expost}{forval}
\newagentvarsm{\epvalloc}{\thealloc}{expost}{forval}
\newagentvarsm{\epqalloc}{\theqalloc}{expost}{forquant}

\newcommand{\theaalloc}{y}
\newagentvarsm{\aalloc}{\theaalloc}{interim}{forbid}
\newagentvarsm{\epaalloc}{\theaalloc}{expost}{forbid}
\newagentvarsm{\apay}{\thepay}{interim}{forbid}
\newagentvarsm{\epapay}{\thepay}{expost}{forbid}

\newagentvarsm{\balloc}{\thealloc}{interim}{forbid}
\newagentvarsm{\bpay}{\thepay}{interim}{forbid}
\newagentvarsm{\epballoc}{\thealloc}{expost}{forbid}
\newagentvarsm{\epbpay}{\thepay}{expost}{forbid}

\newagentvar{\act}{a}
\newagentvar{\bidspace}{A}
\newagentvar{\actspace}{A}

\newagentvarsm{\critval}{\theval}{an}{constrained}
\newagentvarsm{\epvcritval}{\theval}{expost}{constrained}

\newagentvar[\constrained]{\crittype}{\type}
\newagentvar[\constrained]{\critvirt}{\virt}
\newagentvar[\constrained]{\reserve}{\val} % should this be used???
\newagentvar[\constrained]{\bidreserve}{\bid} % should this be used???
\newagentvarsm{\monop}{\val}{an}{optconstrained}
\newagentvar[\constrained]{\monot}{\type}
\newagentvar[\optimized]{\monorev}{\rev}

\newcommand{\thercalloc}{y}
\newagentvarsm{\rcalloc}{\thercalloc}{an}{none}
\newagentvarsm{\optrcalloc}{\thercalloc}{expost}{optimized}
\newagentvarsm{\biddist}{G}{interim}{none}
\newagentvarsm{\biddens}{g}{interim}{none}
\newagentvarsm{\critbid}{\bid}{an}{constrained}
\newagentvarsm{\cbid}{B}{an}{constrained}
\newagentvar[\primedarg]{\wbid}{\bid}
\newagentvar{\gfunc}{\vartheta}
\newagentvar{\block}{C}

\newagentvarsm{\pee}{p}{interim}{none}

\newagentvarsm{\util}{u}{an}{}
\newagentvarsm{\vutil}{u}{interim}{forvalue}
\newagentvarsm{\butil}{u}{interim}{forbid}

\newcommand{\thebid}{b}
\newagentvar{\bid}{\thebid}
\newagentvarsm{\strat}{\thebid}{interim}{}

\newagentvar[\tagged{\text{SD}}]{\sdvirt}{\virt}
\newagentvar[\composed{\tagged{\text{SD}}}{\ironed}]{\sdivirt}{\virt}
\newagentvar[\tagged{\text{MD}}]{\mdvirt}{\virt}
\newagentvar[\composed{\tagged{\text{MD}}}{\ironed}]{\mdivirt}{\virt}

\newagentvar[\ironed]{\iprice}{\price}  %% what for?
\newagentvar[\ironed]{\ival}{\val}  %% what for?

\newagentvar{\ints}{{\cal I}}

\newcommand{\feasibles}{{\cal X}}

\newagentvar{\wal}{w}

\newitemvar{\pos}{j}
\newitemvar{\weight}{w}
\newitemvar[\differentiated]{\mweight}{\weight}
\newitemvar{\udtype}{\type}
\newitemvar[\constrained]{\udcrittype}{\type}
\newitemvar{\udalloc}{\alloc}
\newitemvar{\udprice}{\price}
\newitemvar[\constrained]{\udcalloc}{\qalloc}
\newitemvar[\constrained]{\udcumcalloc}{\cumalloc}

\newagentvar{\mech}{{\cal M}}
\newagentvar[\skew{5}{\hat}]{\cmech}{{\cal M}}
\newagentvar{\alg}{{\cal A}}

\newagentvarsm{\budget}{B}{an}{none}

\newagentvarsm{\rawprice}{\theprice}{expost}{forbid}
\newagentvarsm{\rawalloc}{\thealloc}{expost}{forbid}

\newcommand{\reals}{{\mathbb R}}

\newagentvarsm{\trans}{\sigma}{interim}{none}

\newagentvar{\demandset}{S}

\newcommand{\opt}{^{\star}}
\newcommand{\primed}{^\dagger}

\newagentvar{distout}{w}
\newagentvar{idistout}{\bar{w}}

\newagentvarsm{\epricerev}{\thepricerev}{an}{}
\newagentvarsm{\erev}{\therev}{an}{}
\newagentvarsm{\emarg}{\therev}{an}{differentiated}
\newagentvarsm{\epricemarg}{\thepricerev}{an}{differentiated}
\newagentvarsm{\efalloc}{\theqalloc}{interim}{none}
\newagentvarsm{\evirt}{\thevirt}{an}{none}

\newagentvar{\gap}{\delta}

\newagentvar[\tilde]{\ualloc}{\alloc}
\newagentvar[\tilde]{\uprice}{\price}

\newcommand{\given}{\,\mid\,}

\newcommand{\prob}[2][]{\text{\bf Pr}\ifthenelse{\not\equal{}{#1}}{_{#1}}{}\!\left[{\def\givenn{\middle|}#2}\right]}
\newcommand{\expect}[2][]{\text{\bf E}\ifthenelse{\not\equal{}{#1}}{_{#1}}{}\!\left[{\def\givenn{\middle|}#2}\right]}

\newcommand{\tparen}{\big}
\newcommand{\tprob}[2][]{\text{\bf Pr}\ifthenelse{\not\equal{}{#1}}{_{#1}}{}\tparen[{\def\given{\tparen|}#2}\tparen]}
\newcommand{\texpect}[2][]{\text{\bf E}\ifthenelse{\not\equal{}{#1}}{_{#1}}{}\tparen[{\def\given{\tparen|}#2}\tparen]}

\newcommand{\sprob}[2][]{\text{\bf Pr}\ifthenelse{\not\equal{}{#1}}{_{#1}}{}[#2]}
\newcommand{\sexpect}[2][]{\text{\bf E}\ifthenelse{\not\equal{}{#1}}{_{#1}}{}[#2]}

\title{\vspace{-4ex}Online Bipartite Matching via Smoothness}
\author{Jason Hartline}

\date{}

\let\section=\paragraph

\begin{document}

\maketitle

\vspace{-2ex}

The online bipartite matching problem has offline buyers desiring to
be matched to online items.  A bipartite graph governs which buyers
are interested in which items.  The RANKING algorithm of
\citet{KVV-90} assigns the buyers a uniform random priority and then,
as items arrive and their edges to remaining unmatched buyers become
known, matches each item to the interested buyer with the highest
priority (if any exists).  They give a tight analysis of the RANKING
algorithm showing that it is an $e/(e-1)$ approximation; i.e., it
matches a 1.58 fraction of the number of matched buyers in the optimal
offline matching.

Recently, \citet{EFFS-21} reproved the upperbound using an economic
analysis.  They imagine that the buyers all have value 1, they bid
i.i.d.\ from a specific distribution, and consider the greedy-by-bid
mechanism, where each item is matched when it arrives to the feasible
buyer with the highest bid.  They show that for these bids that
welfare obtained (which is equal to the number of matched buyers) is a
1.58 approximation to the optimal welfare.\footnote{Technically, this
discussion has swapped the role of buyers and items from
\citet{EFFS-21}.}  As the outcome produced by this bidding process is
identical to that of RANKING, the same approximation bound holds for
it.

In fact, the analysis of online bipartite matching of \citet{EFFS-21} is a
``smoothness proof'' where the bid distribution is the same as the one
defined in the smoothness analysis of \citet{ST-13}.  Moreover, it can
be interpreted as combining $\lambda = 1-1/e$ {\em value covering},
which is known to hold for single-dimensional buyers and randomized
auctions, and $\mu = 1$ {\em revenue covering} \citep{HHT-14}, which
is shown to hold for the greedy-by-bid winner-pays-bid online matching
mechanism.  Note that value covering is a property of
single-dimensional buyers in winner-pays-bid mechanisms and has
nothing to do with the underlying feasibility setting.  Thus, the
essential new consequence of \citet{EFFS-21} is that online bipartite
matching is $\mu=1$ revenue covered.  A number of old and new
observations follow from this perspective.

The remainder of this brief note is organized as follows.  First, the
definition of revenue covering is reviewed along with the main
analysis that shows that the greedy-by-bid winner-pays-bid online bipartite
matching mechanism is $\mu=1$ revenue covered.  Second, value covering
is reviewed with a discussion of the inherent optimality of oblivious
bidding strategies.  Third, the approach is applied to the analysis of
the RANKING algorithm.  This brief note concludes with a discussion of
old and new consequences of this reorganization of the analysis.

\paragraph{Revenue Covering.}
A single dimensional buyer $i$ has a value $\vali$ and linear utility
$\utili = \vali \alloci - \pricei$ for obtaining a good or service
with probability $\alloci$ and making payment $\payi$.  A feasibility
constraint is defined by $\feasibles \subset [0,1]^n$.  The optimal
welfare is the maximum over feasible allocations $\allocs =
(\alloci[1],\ldots,\alloci[n]) \in \feasibles$ of $\sum_i \vali
\alloci$.  A winner-pays-bid mechanism is fully specified by its ex
post bid allocation rule $\ballocs : \reals^n \to [0,1]^n$ and on input bids $\bids = (\bidi[1],\ldots,\bidi[n])$ its
revenue is $\sum_i \bidi\,\balloci(\bids)$.  Any randomized
winner-pays-bid mechanism and fixed bids of other buyers induces a
distribution over critical bids $\critbidi$ where buyer $i$ wins when
her bid $\bidi$ exceeds $\critbidi$.  A winner-pays-bid mechanism
$\ballocs$ is $\mu \geq 1$ revenue covered if on all profiles of bids
$\bids$ and all other feasible allocations $\allocs \in \feasibles$,
\begin{align}
  \label{eq:rc}
  \mu \sum\nolimits_i \bidi\, \balloci(\bids) &\geq \sum\nolimits_i \alloci \critbidi.
\end{align}
Here the left-hand side is $\mu$ times the revenue of the mechanism
and the right-hand side is the ``welfare of critical bids'' according
to feasible allocation $\allocs \in \feasibles$.

\paragraph{Revenue Covering of Online Bipartite Matching.}

The main observation of this note is reinterpreting the analysis of
\citet{EFFS-21} as showing that the greedy online matching mechanism, which assigns each
item when it arrives to the remaining interested buyer with the highest bid, is $\mu=1$
revenue covered.

\begin{theorem}
  \label{t:revenue-covering}
  For any fixed ordering of items, greedy-by-bid online matching is
  $\mu=1$ revenue covered.
\end{theorem}

\begin{proof}
For buyers $A$, items $B$, and the bipartite graph given by $(A,B,E)$;
consider any feasible matching $M \subset E$ (not necessarily the
matching produced by an online allocation). Let $\ballocs$ be the
allocation rule of online bipartite matching for offline buyers, any
fixed arrival order of items, and where each item on arrival is sold
to the highest remaining bidder that is feasible according to $E$.
Fix bids $\bids$ and critical bids $\critbids$ (induced by
$\ballocs$).  Denote by $r_j$ the revenue from item $j$ (equal to the bid of
the buyer who is matched to $j$ or equal to 0 if $j$ is unmatched).  To be shown:
\begin{align}
  \sum\nolimits_i \bidi \, \balloci(\bidi) \geq \sum\nolimits_{j : (i,j) \in M}r_j \geq
  \sum\nolimits_{i : (i,j) \in M} \critbidi.
\end{align}
Notice that the right-hand side is $\sum_i \alloci \critbidi$ for
feasible $\allocs$ defined with $\alloci$ as an indicator variable
for buyer $i$ being matched in $M$.  Thus, online matching is $\mu=1$
revenue covered.

The first inequality is true because the sum of winning bids is the
revenue which is at least the sum of revenues from a subset of items.

For the second inequality, imagine buyer $i$ does not participate in
the market and define $q$ to be the price that item $j$ would be sold
at in the online matching (or 0 if item $j$ remains unsold).  Observe
two results.  First, it can be seen by induction that when buyer $i$
is in the market the price fetched by item $j$ is only higher, i.e.,
$r_j \geq q$.  Item $j$ is sold to the highest bidder remaining when
it arrives.  Adding a buyer to the consideration set of buyers when an
item arrives increases the item's price and either the original buyer
buys and this new buyer is available for later items or the new buyer
buys and the original buyer is available for later items.  Second, if
buyer $i$ bids more than $q$ then buyer $i$ will certainly be matched,
because if unmatched for all other items, buyer $i$ would be matched
to $j$.  Thus, the minimum bid that $i$ can make and still be matched
by $\ballocs$ is $\critbidi \leq q$.  Combining these two observations
and summing over all $(i,j) \in M$ the inequalities $r_j \geq q \geq
\critbidi$ gives the second inequality.
\end{proof}

\paragraph{Best response, oblivious response, and value covering.}

Consider a buyer with value $\val$ bidding in a winner-pays-bid
auction and any, not necessarily equilibrium, distribution of opposing
bids.  This buyer wins whenever her bid exceeds a potentially
randomized critical bid $\critbid$ that is induced by the distribution
of bids of other buyers and the auction rules.  The {\em value
  covering} bound shows that either the best-response utility $\util$
of the buyer or their expected critical bid $\critbid$ is large in
comparison to their value $\val$:
\begin{align}
  \label{eq:value-covering}
  \sexpect{\util} + \sexpect{\critbid} \geq \lambda\,\val
\end{align}
If inequality~\eqref{eq:value-covering} holds for {\em oblivious
  response}, i.e., a fixed randomized bidding strategy, then it also
holds for {\em best response}, i.e., the best bid for the distribution
of $\critbid$.  \citet{ST-13} identify an oblivious response that
yields the tight bound $\lambda = 1-1/e$.  Moreover, in worst case oblivious
response and best response give the same $\lambda$.  Inequality~\eqref{eq:value-covering} can be seen as defining a zero-sum game
between the buyer (playing $\bid$) and nature (playing $\critbid$)
where $\lambda$ is the value of the game.  An important property of
zero-sum games is that each player has an oblivious randomized
strategy that guarantees them at least the value of the game for any
strategy of the opponent (and is tight in worst case).  The oblivious strategy
that guarantees the $\lambda = 1-1/e$ value is to bid according to the
density function $g(z) = 1/(\val - z)$ on interval $[0,(1-1/e)\,\val]$.
The utility from such a bid for any fixed $\critbid$ is (integrating
the winner-pays-bid utility of $\val - z$ over winning bids $z \geq
\critbid$):
\begin{align*}
  \expect{\util} &= \int_{\critbid}^{(1-1/e)\,\val}(\val - z)\,g(z)\,dz =
  \int_{\critbid}^{(1-1/e)\,\val} 1\,dz =
      (1-1/e)\,\val - \critbid.
\end{align*}
For fixed $\critbid$ the equation is linear, thus it also holds in
expectation over randomized $\critbid$.  Rearranging gives
inequality~\eqref{eq:value-covering} with $\lambda = 1-1/e$.

\citet{HHT-14} show that value covering and revenue covering combine
to give a bound on the welfare in equilibrium. Value covering implies
$\sexpect{\utili} + \sexpect{\critvali}\,\alloci \geq \lambda\,\vali\,\alloci$
for each buyer $i$ and any $\alloci \in [0,1]$.  Taking $\allocs$ as the
welfare optimal allocation, summing over all buyers, applying $\mu$
revenue covering, and rearranging with $\mu \geq 1$ gives the theorem (\Cref{t:hht}, below).

\begin{theorem}[\citealp{HHT-14}]
  \label{t:hht}
  The outcome from buyers playing $\lambda$-value-covered strategies
  in a $\mu$-revenue-covered mechanism is a $\mu/\lambda$
  approximation to the optimal welfare.
\end{theorem}

\paragraph{Online Algorithms.}

The value covering analysis is algorithmic in the sense that it gives
an algorithm the buyer could follow that guarantees its bound.  In any
equilibrium, the best response of the buyer is only better.  While the
best response of a buyer is based on beliefs that the buyer
possesses about the environment and other buyers' strategies, the
bidding strategy in the analysis is independent of these details.
Combining the oblivious bidding strategy with a mechanism gives an
algorithm that obtains the same welfare guarantee as is obtained from
equilibrium.

\citet{KVV-90} consider the unweighted bipartite matching; their
theorem can be rederived from $\lambda=1-1/e$ value covering, $\mu =
1$ revenue covering, and \Cref{t:hht} with valuation profile $\vals =
(1,\ldots,1)$.  The oblivious bidding strategy combined with the
greedy-by-bid algorithm for online bipartite matching gives an
identical outcome to the RANKING algorithm.

\begin{theorem}
  The RANKING algorithm (equivalently, oblivious bidding with
  greedy-by-bid) gives an $e/(e-1)$ approximation to online bipartite
  matching.
\end{theorem}

\paragraph{Discussion.}

Separating the proof of $\mu=1$ revenue covering from the calculus of
the $\lambda=1-1/e$ value covering of \citet{ST-13} clarifies the
analysis of \citet{EFFS-21}.  The $\mu=1$ revenue covering result
affords other conclusions as well.  For example, when buyers have
asymmetric values, oblivious bidding with greedy-by-bid continues to
give an $e/(e-1)$ approximation.  These bounds that hold for the
RANKING algorithm also hold in equilibrium of the winner-pays-bid
greedy-by-bid online matching algorithm (when buyers have beliefs over
the arrival order of items).  It is also known that in deterministic
environments, e.g., where the arrival order is known to the buyers,
that best response in winner-pays-bid mechanisms is $\lambda=1$ value
covered.  Thus, the pure Nash equilibria of the greedy-by-bid
mechanism are perfectly efficient.  In contrast, the greedy algorithm
without randomization or equilibrium bidding is a 2-approximation. (In
fact, this well-known greedy 2-approximation result can be seen as a
consequence of $\mu=1$ revenue covering and $\lambda = 1/2$ value
covering that is obtained by the oblivious ``bid half of value''
bidding strategy.)

\citet{MSVV-05} applied online matching to the problem of selling
advertisements on keyword searches, a.k.a., the AdWords problem.
Their analysis focuses on the algorithmic challenge, explicitly
leaving the problem of advertiser incentives for future work.  In
contrast, the analysis above suggests that it might be better to
simply run the greedy-by-bid online matching algorithm and charge the
winners their bids; the strategic bidding of the advertisers ensures that
the welfare obtained is good.

\bibliographystyle{apalike}
\bibliography{references}

\end{document}